\documentclass[submission,copyright,creativecommons]{eptcs}

\providecommand{\event}{TARK 2023} % Name of the event you are submitting to

\usepackage{iftex}
\ifpdf
  \usepackage{underscore}         % Only needed if you use pdflatex.
  \usepackage[T1]{fontenc}        % Recommended with pdflatex
\else
  \usepackage{breakurl}           % Not needed if you use pdflatex only.
\fi

\usepackage[english]{babel}

\title{Maximizing Social Welfare in Score-Based\\ Social Distance Games}

\author{%
	Robert Ganian
	\institute{TU Wien\\Vienna, Austria}
	\email{rganian@gmail.com}
\and
	Thekla Hamm
	\institute{Utrecht University\\Utrecht, Netherlands}
	\email{thekla.hamm@gmail.com}
\and
	Dušan Knop
	\institute{Czech Technical University in Prague\\Prague, Czech Republic}
	\email{dusan.knop@fit.cvut.cz}
\and
	Sanjukta Roy
	\institute{Penn State University\\Pennsylvania, USA}
	\email{sanjukta@psu.edu}
\and
	Šimon Schierreich
	\institute{Czech Technical University in Prague\\Prague, Czech Republic}
	\email{schiesim@fit.cvut.cz}
\and
	Ondřej Suchý
	\institute{Czech Technical University in Prague\\Prague, Czech Republic}
	\email{ondrej.suchy@fit.cvut.cz}
}

\def\titlerunning{Maximizing Social Welfare in Score-Based Social Distance Games}

\def\authorrunning{R. Ganian, T. Hamm, D. Knop, S. Roy, Š. Schierreich \& O. Suchý}

\usepackage{amsmath}
\usepackage{amssymb}
\usepackage{amsthm}
\usepackage{xspace}
\usepackage{bm}
\usepackage[textwidth=16mm,textsize=footnotesize,disable]{todonotes}
\usepackage{tikz}
\usetikzlibrary{hobby,backgrounds,positioning,calc,trees}
\usepackage{boxedminipage}
\usepackage{cleveref}
\usepackage[inline]{enumitem}
\usepackage{subcaption}
\usepackage[mathlines,displaymath]{lineno}

\newcommand{\N}{\ensuremath{\mathbb{N}}}

\newcommand{\Z}{\ensuremath{\mathbb{Z}}}

\newcommand{\cdiam}{\ensuremath{{\delta}}}

\newcommand{\bigoh}{\ensuremath{\mathcal{O}}}
\newcommand{\Oh}[1]{\ensuremath{{\bigoh(#1)}}}

\newcommand{\cc}[1]{{\mbox{\textnormal{\textsf{#1}}}}\xspace}  %% Complexity class
\newcommand{\FPT}{\cc{FPT}}
\newcommand{\XP}{\cc{XP}}
\newcommand{\NP}{{\cc{NP}}}

\newcommand{\NPh}{\NP-hard\xspace}

\newcommand{\YES}{\cc{Yes}}
\newcommand{\NO}{\cc{No}}
\newcommand{\Yes}{\YES}
\newcommand{\No}{\NO}
\newcommand{\YesI}{\YES-instance\xspace}

\newtheorem{theorem}{Theorem}
\crefname{theorem}{Theorem}{Theorems}

\crefname{observation}{Observation}{Observations}
\newtheorem{lemma}[theorem]{Lemma}
\crefname{lemma}{Lemma}{Lemmas}
\newtheorem{corollary}[theorem]{Corollary}
\crefname{corollary}{Corollary}{Corollaries}
\newtheorem{proposition}[theorem]{Proposition}
\crefname{proposition}{Proposition}{Propositions}

\crefname{conjecture}{Conjecture}{Conjectures}

\crefname{claim}{Claim}{Claims}

\theoremstyle{remark}
\crefname{example}{Example}{Examples}

\newcommand{\DGs}{\textsc{SDG}\xspace}
\newcommand{\DGNS}{\val-\textsc{SDG-WF-Nash}\xspace}
\newcommand{\DGIR}{\val-\textsc{SDG-WF-IR}\xspace}
\newcommand{\DGWF}{\val-\textsc{SDG-WF}\xspace}

\newcommand{\dist}{\ensuremath{\operatorname{dist}}} % Distance

\newcommand{\val}{\ensuremath{\vec{\operatorname{s}}}\xspace} % Scoring vector
\newcommand{\maxval}{s_1\xspace}%maximum positive score
\newcommand{\util}{\ensuremath{\operatorname{u}}\xspace} % Utility function
\newcommand{\SW}{\ensuremath{\operatorname{SW}}\xspace} % Social welfare function
\newcommand{\topo}{\ensuremath{\mathsf{T}}\xspace} % Social welfare function

\newcommand{\tw}{\ensuremath{\operatorname{tw}}\xspace} % Tree-width
\newcommand{\sz}{\ensuremath{\operatorname{sz}}\xspace} % maximum size of coalition
\newcommand{\vc}{\ensuremath{\operatorname{vc}}\xspace} % Vertex cover
\begin{document}

\maketitle
%\linenumbers

\begin{abstract}
Social distance games have been extensively studied as a coalition formation model where the utilities of agents in each coalition were captured using a utility function $\util$ that took into account distances in a given social network. In this paper, we consider a non-normalized score-based definition of social distance games where the utility function $u^{\val}$ depends on a generic scoring vector $\val$, which may be customized to match the specifics of each individual application scenario.

As our main technical contribution, we establish the tractability of computing a welfare-ma\-xi\-mi\-zing partitioning of the agents into coalitions on tree-like networks, for every score-based function $u^{\val}$. We provide more efficient algorithms when dealing with specific choices of $u^{\val}$ or simpler networks, and also extend all of these results to computing coalitions that are Nash stable or individually rational.
We view these results as a further strong indication of the usefulness of the proposed score-based utility function: even on very simple networks, the problem of computing a welfare-maximizing partitioning into coalitions remains open for the originally considered canonical function $\util$.
\end{abstract}

%%%%%%%%%%%%%%%%%%%%%%%%%%%%%%%%%%%%%%%%%%%%%%%%%%%%%%%%%%%%%%
\section{Introduction}

Coalition formation is a central research direction within the fields of algorithmic game theory and computational social choice. While there are many different scenarios where agents aggregate into coalitions, a pervasive property of such coalitions is that the participating agents exhibit \emph{homophily}, meaning that they prefer to be in coalitions with other agents which are similar to them. It was this observation that motivated Br{\^{a}}nzei and Larson to introduce the notion of \emph{social distance games} (SDG) as a basic model capturing the homophilic behavior of agents in a social network~\cite{BranzeiL2011}.

Br{\^{a}}nzei and Larson's SDG model consisted of a graph $G=(V,E)$ representing the social network, with $V$ being the agents and $E$ representing direct relationships or connections between the agents. To capture the utility of an agent $v$ in a coalition $C\subseteq V$, the model considered a single function: $u(v,C)=\frac{1}{|C|}\cdot \sum\nolimits_{w\in C\setminus \{v\}}\frac{1}{d_C(v,w)}$ where $d_C(v,w)$ is the distance between $v$ and $w$ inside $C$.

Social distance games with the aforementioned utility function $\util$ have been the focus of extensive study to date, with a number of research papers specifically targeting algorithmic and complexity-theoretic aspects of forming coalitions with maximum social welfare~\cite{BalliuFMO17,BalliuFMO19,BalliuFMO22,KaklamanisKP18}.
Very recently, Flammini et al.~\cite{FlamminiKOV2020,FlamminiKOV2021} considered a generalization of $\util$ via an adaptive real-valued scoring vector which weights the contributions to an agent's utility according to the distances of other agents in the coalition, and studied the price of anarchy and stability for non-negative scoring vectors.
However, research to date has not revealed any polynomially tractable fragments for the problem of computing coalition structures with maximum social welfare (with or without stability-based restrictions on the behavior of individual agents), except for the trivial cases of complete (bipartite) graphs~\cite{BranzeiL2011} and trees~\cite{OkuboHO19}.

\paragraph{Our Contribution.}
The undisputable appeal of having an adaptive scoring vector---as opposed to using a single canonical utility function $\util$---lies in the fact that it allows us to capture many different scenarios with different dynamics of coalition formation. However, it would also be useful for such a model to be able to assign negative scores to agents at certain (larger) distances in a coalition.
For instance, guests at a gala event may be keen to accept the presence of friends-of-friends (i.e., agents at distance $2$) at a table, while friends-of-friends may be less welcome in private user groups on social networks, and the presence of complete strangers in some scenarios may even be socially unacceptable.

Here, we propose the study of social distance games with a family of highly generic non-normalized score-based utility functions.
Our aim here is twofold. First of all, these should allow us to better capture situations where agents at larger distances are unwelcome or even unacceptable for other agents.
At the same time, we also want to obtain algorithms capable of computing welfare-maximizing coalition structures in such general settings, at least on well-structured networks.

Our model considers a graph $G$ accompanied with an integer-valued, fixed but adaptive \emph{scoring vector} $\val$ which captures how accepting agents are towards other agents based on their pairwise distance.\footnote{Formal definitions are provided in the Preliminaries.} The utility function $u^{\val}(v,C)$ for an agent $v$ in coalition $C$ is then simply defined as $u^{\val}(v,C)=\sum\nolimits_{w\in C\setminus \{v\}} \val(d_C(v,w))$; we explicitly remark that, unlike previous models, this is not normalized with respect to the coalition size. As one possible example, a scoring vector of $(1,0,-1)$ could be used in scenarios where agents are welcoming towards friends, indifferent to friends-of-friends, slightly unhappy about friends-of-friends-of-friends (i.e., agents at distance $3$), and unwilling to group up with agents who are at distance greater than $3$ in $G$.
A concrete example which also illustrates the differences to previous SDG models is provided in \Cref{fig:example}.

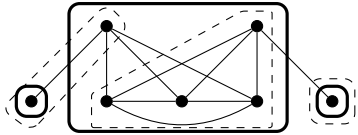
\begin{figure}
	\begin{center}
		\begin{tikzpicture}
			[
			every node/.style={draw, fill=black, shape=circle, inner sep=0pt, text width=1.5mm, align=center, label distance=1mm}
			]
			\node (x1) at  (-2,0) {};
			\node (a1) at  (-1,0) {};
			\node (a2) at  (0,0) {};
			\node (a3) at  (1,0) {};
			\node (y1) at  (2,0) {};

			\node (x) at (-1,1) {};
			\node (y) at (1,1) {};

			\foreach \i in {a1,a2,a3}
			{
				\draw (x) to (\i) to (y);
			}
			\draw (x) to (x1);
			\draw (y) to (y1);
			\draw (a1) to (a2) to (a3) to[bend left] (a1);

			\draw[very thick,rounded corners] (-1.5,-0.4) rectangle (1.4,1.3);
			\draw[very thick,rounded corners] (-2.2,-0.2) rectangle (-1.8,0.2);
			\draw[very thick,rounded corners] (2.2,-0.2) rectangle (1.8,0.2);

			\draw[dashed,rounded corners] (-2,-0.4) -- (-0.7,1) -- (-1.,1.3) -- (-2.4,-0.1) -- cycle;
			\draw[dashed,rounded corners=0.3mm] (-1.2,-0.35) -- (1.2,-0.35) -- (1.2,1.2) -- (0.8,1.2) -- (-1.2,0.1) -- cycle;
			\draw[dashed,rounded corners] (2.3,-0.3) rectangle (1.7,0.3);
		\end{tikzpicture}
		\caption{A social network illustrating the difference of maximising social welfare in our model compared to previous SDG models. (1) In Br{\^{a}}nzei and Larson's SDG model, the welfare-maximum outcome is the grand coalition. (2) A welfare-maximum outcome in the normalized model of Flammini et al.\ with a scoring vector of $(1,0,0,0)$ is marked with dashed lines, while the same scoring vector in our non-normalized model produces the grand coalition. (3) A scoring vector of \(\val = (1,0,-1)\) in our model produces the welfare-maximizing outcome marked with bold lines, with a welfare of $18$. (4) A `less welcoming' scoring vector of \(\val = (1,-3)\) leads to the welfare maximizing dash-circled partition with a welfare of \(14\) (compared to only \(12\) for the bold-circled one).
			\label{fig:example}}
	\end{center}
\end{figure}

While non-normalized scoring functions have not previously been considered for social distance games, we view them a natural way of modeling agent utilities; in fact, similar ideas have been successfully used in models for a variety of other phenomena including, e.g., committee voting~\cite{ElkindI15}, resource allocation~\cite{BrandtCELP16,BouveretL08}
and Bayesian network structure learning~\cite{GanianK21,OrdyniakS13}.
Crucially, it is not difficult to observe that many of the properties originally established by Br{\^{a}}nzei and Larson for SDGs also hold for our non-normalized score-based model with every choice of $\val$, such as the small-world property~\cite{BranzeiL2011,jackson2008social} and the property that adding an agent with a close (distant) connection to a coalition positively (negatively) impacts the utilities of agents~\cite{BranzeiL2011}.
In addition, the proposed model can also directly capture the notion of \emph{enemy aversion} with symmetric preferences~\cite{BarrotY19,OhtaBISY17} by setting $\val=(1)$.

Aside from the above, a notable benefit of the proposed model lies on the complexity-theoretic side of things. Indeed, a natural question that arises in the context of SDG is whether we can compute an outcome---a partitioning of the agents into coalitions---which maximizes the social welfare (defined as the sum of the utilities of all agents in the network). This question has been studied in several contexts, and depending on the setting one may also require the resulting coalitions to be stable under \emph{individual rationality} (meaning that agents will not remain in coalitions if they have negative utility) or \emph{Nash stability} (meaning that agents may leave to join a different coalition if it would improve their utility). But in spite of the significant advances in algorithmic aspects of other coalition formation problems in recent years~\cite{BoehmerE20,BoehmerE20b,ChenGH20,GanianHKSS22},
we lack any efficient algorithm capable of producing such a welfare-optimal partitioning when using the utility function $\util$ even for the simplest types of networks.

To be more precise, when viewed through the refined lens of \emph{parameterized complexity}~\cite{CyganFKLMPPS2015,DowneyF13} that has recently become a go-to paradigm for such complexity-theoretic analysis, no tractable fragments of the problem are known. More precisely, the problem of computing a welfare-maximizing outcome under any of the previously considered models is not even known to admit an \XP\ algorithm when parameterized by the minimum size of a vertex cover in the social network $G$---implying a significant gap towards potential fixed-parameter tractability. This means that the complexity of welfare-maximization under previous models remains wide open even under the strongest non-trivializing restriction of the network.

As our main technical contribution, we show that non-normalized score-based utility functions do not suffer from this drawback and can in fact be computed efficiently under fairly mild restrictions on $G$. Indeed, as our first algorithmic result we obtain an \XP\ algorithm that computes a welfare-maximizing partitioning of the agents into coalitions parameterized by the treewidth of $G$, and we strengthen this algorithm to also handle additional restrictions on the coalitions in terms of individual rationality or Nash stability.
As with numerous treewidth-based algorithms, we achieve this result via leaf-to-root dynamic programming along a tree-decomposition. However, the records we keep during the dynamic program are highly non-trivial and require an advanced branching step to correctly pre-computed the distances in the stored records.
We remark that considering networks of small treewidth is motivated not only by the fundamental nature of this structural graph measure, but also by the fact that many real-world networks exhibit bounded treewidth~\cite{ManiuSJ19}.

In the next part of our investigation, we show that when dealing with simple scoring functions or bounded-degree networks, these results can be improved to fixed-parameter algorithms for welfare-maximization (including the cases where we require the coalitions to be individually rational or Nash stable). This is achieved by combining structural insights into the behavior of such coalitions with a different dynamic programming approach. Furthermore, we also use an entirely different technique based on quadratic programming to establish the fixed-parameter tractability of all 3 problems under consideration w.r.t.\ the minimum size of a vertex cover in $G$. Finally, we conclude with some interesting generalizations and special cases of our model and provide some preliminary results in these directions.

%%%%%%%%%%%%%%%%%%%%%%%%%%%%%%%%%%%%%%%%%%%%%%%%%%%%%%%%%%%%%%

\section{Preliminaries}
We use $\N$ to denote the set of natural numbers, i.e., positive integers, and $\Z$ for the set of integers. For $i\in \N$, we let $[i]= \{1,\ldots,i\}$ and ${[i]_0 = [i] \cup \{0\}}$.
We assume basic familiarity with graph-theoretic terminology~\cite{Diestel17}.

\paragraph{Social Distance Games.}
A \emph{social distance game} (SDG) consists of a set $N = \{1,\ldots,n\}$ of \emph{agents}, a simple undirected graph $G=(N,E)$ over the set of agents called a \emph{social network}, and a non-increasing \emph{scoring vector} $\val=(s_1,\dots,s_{\cdiam})$ where
\begin{enumerate*}[label=\alph*)]
	\item for each $a\in [\cdiam]$, $s_a\in \Z$ and
	\item for each $a\in [\cdiam-1]$, $s_{a+1}\leq s_a$.
\end{enumerate*}

In some cases, it will be useful to treat $\val$ as a function from $\N$ rather than a vector; to this end, we set $\val(a)=s_a$ for each $a\leq \cdiam$ and $\val(a)=-\infty$ when $a>\cdiam$.
The value ``$-\infty$'' here represents an inadmissible outcome, and formally we set $-\infty+z=-\infty$ and $-\infty<z$ for each $z\in \Z$.

A \emph{coalition} is a subset $C\subseteq N$, and an outcome is a partitioning $\Pi=(C_1,\dots,C_\ell)$ of $N$ into coalitions; formally,
$\bigcup_{i=1}^\ell C_i = N$, every $C_i\in \Pi$ is a coalition, and all coalitions in $\Pi$ are pairwise disjoint. We use~$\Pi_i$ to denote the coalition the agent $i\in N$ is part of in the outcome $\Pi$.
The \emph{utility} of an agent $i\in N$ for a coalition $\Pi_i\in\Pi$ is
\[
\util^{\val}(i,\Pi_i) = \sum_{j\in \Pi_i\setminus\{i\}} \val(\dist_{\Pi_i}(i,j)),
\]
where $\dist_{\Pi_i}(i,j)$ is the length of a shortest path between $i$ and $j$ in the graph $G[\Pi_i]$, i.e., the subgraph of~$G$ induced on the agents of $\Pi_i$. We explicitly note that if $\Pi_i$ is a singleton coalition then $\util^{\val}(i,\Pi_i)=0$. Moreover, in line with previous work~\cite{BranzeiL2011} we set $\dist_{\Pi_i}(i,j):=+\infty$ if there is no $i$-$j$ path in $G[\Pi_i]$, meaning that $\util^{\val}(i,\Pi_i)=-\infty$ whenever $G[\Pi_i]$ is not connected.

For brevity, we drop the superscript from $u^{\val}$ whenever the scoring vector \val\ is clear from the context. To measure the satisfaction of the agents with a given outcome, we use the well-known notation of \emph{social welfare}, which is the total utility of all agents for an outcome $\Pi$, that is,
\[
\SW^{\val}(\Pi) = \sum_{i\in N} \util^{\val}(i,\Pi_i).
\]
Here, too, we drop the superscript specifying the scoring vector whenever it is clear from the context.

We assume that all our agents are selfish, behave strategically, and their aim is to maximize their utility. To do so, they can perform \emph{deviations} from the current outcome $\Pi$. We say that $\Pi$ admits an \emph{IR-deviation} if there is an agent $i\in N$ such that $\util(i,C) < 0$; in other words, agent $i$ prefers to be in a singleton coalition over its current coalition. If no agent admits an IR-deviation, the outcome is called \emph{individually rational} (IR). We say that $\Pi$ admits an \emph{NS-deviation} if there is an agent $i$ and a coalition $C\in \Pi\cup \{\emptyset\}$ such that $\util(i,C\cup\{i\}) > \util(i,\Pi_i)$. $\Pi$~is called \emph{Nash stable} (NS) if no agent admits an NS-deviation.
We remark that other notions of stability exist in the literature~\cite[Chapter 15]{BrandtCELP16}, but Nash stability and individual rationality are the most basic notions used for stability based on individual choice~\cite{Karakaya11,SungD07}.

Having described all the components in our score-based SDG model, we are now ready to formalize the three classes of problems considered in this paper. We note that even though these are stated as decision problems for complexity-theoretic reasons, each of our algorithms for these problems can also output a suitable outcome as a witness. For an arbitrary fixed scoring vector $\val$, we define:

\begin{center}
	\begin{boxedminipage}{0.98 \columnwidth}
		\DGWF\\[5pt]
		\begin{tabular}{l p{0.78 \columnwidth}}
			Input: & A social network $G=(N,E)$, desired welfare $b \in \N$.\\
			Question: \hspace{-0.4cm} & Does the distance game given by $G$ and $\val$ admit an outcome with social welfare at least~$b$?
		\end{tabular}
	\end{boxedminipage}
\end{center}

\DGIR\ and \DGNS\ are then defined analogously, but with the additional condition that the outcome must be individually rational or Nash stable, respectively.

We remark that for each of the three problems, one may assume w.l.o.g.\ that $s_1>0$; otherwise the trivial outcome consisting of $|N|$ singleton coalitions is both welfare-optimal and stable.
Moreover, without loss of generality we assume $G$ to be connected since an optimal outcome for a disconnected graph $G$ can be obtained as a union of optimal outcomes in each connected component of $G$.

The last remark we provide to the definition of our model is that it trivially also supports the
well-known \emph{small world} property~\cite{jackson2008social} that has been extensively studied on social networks.
In their original work on SDGs, Br{\^{a}}nzei and Larson showed that their model exhibits the small world property by establishing a diameter bound of $14$ in each coalition in a so-called \emph{core partition}~\cite{BranzeiL2011}.
Here, we observe that for each choice of $\val$, a welfare-maximizing coalition will always have diameter at most $\cdiam$.

\paragraph{Parameterized Complexity.} The \emph{parameterized complexity} framework~\cite{CyganFKLMPPS2015,DowneyF13} provides the ideal tools for the fine-grained analysis of computational problems which are \NPh\ and hence intractable from the perspective of classical complexity theory. Within this framework, we analyze the running times of algorithms not only with respect to the input size $n$, but also with respect to a numerical parameter $k\in\N$ that describes a well-defined structural property of the instance; the central question is then whether the superpolynomial component of the running time can be confined by a function of this parameter alone.

The most favorable complexity class in this respect is \FPT (short for ``fixed-parameter tractable'') and contains all problems solvable in $f(k)\cdot n^\Oh{1}$ time, where $f$ is a computable function. Algorithms with this running time are called \emph{fixed-parameter algorithms}. A less favorable, but still positive, outcome is an algorithm with running time of the form~$n^{f(k)}$; problems admitting algorithms with such running times belong to the class \XP.

\paragraph{Structural Parameters.} Let $G=(V,E)$ be a graph. A set $U\subseteq V$ is a \emph{vertex cover} if for every edge $e\in E$ it holds that $U\cap e \not= \emptyset$. The \emph{vertex cover number} of $G$, denoted $\vc(G)$, is the minimum size of a vertex cover of $G$.
A \emph{nice tree-decomposition} of $G$ is a pair $(\mathcal{T},\beta)$, where $\mathcal{T}$ is a tree rooted at a node $r\in V(\mathcal{T})$, $\beta\colon V(\mathcal{T})\to 2^{V}$ is a function assigning each node $x$ of $\mathcal{T}$ its \emph{bag}, and the following conditions hold:
\begin{itemize}
	\item for every edge $\{u,v\}\in E(G)$ there is a node $x\in V(\mathcal{T})$ such that $u,v\in\beta(x)$,
	\item for every vertex $v\in V$, the set of nodes $x$ with $v\in\beta(x)$ induces a connected subtree of $\mathcal{T}$,
	\item $|\beta(r)|=|\beta(x)| = 0$ for every \emph{leaf} $x\in V(\mathcal{T})$, and
	\item there are only tree kinds of internal nodes in $\mathcal{T}$:
	\begin{itemize}
		\item $x$ is an \emph{introduce node} if it has exactly one child $y$ such that $\beta(x) = \beta(y)\cup\{v\}$ for some ${v\notin\beta(y)}$,
		\item $x$ is a \emph{join node} if it has exactly two children $y$ and $z$ such that $\beta(x) = \beta(y) = \beta(z)$, or
		\item $x$ is a \emph{forget node} if it has exactly one child $y$ such that $\beta(x) = \beta(y)\setminus\{v\}$ for some $v\in\beta(y)$.
	\end{itemize}
\end{itemize}
The \emph{width} of a nice tree-decomposition $(\mathcal{T},\beta)$ is $\max_{x\in V(\mathcal{T})} |\beta(x)|-1$, and the treewidth $\tw(G)$ of a graph $G$ is the minimum width of a nice tree-decomposition of $G$. Given a nice tree-decomposition and a node~$x$, we denote by $G^x$ the subgraph induced by the set $V^x = \bigcup_{y\text{ is a descendant of }x}\beta(y)$, where we suppose that~$x$ is a descendant of itself.
It is well-known that optimal nice tree-decompositions can be computed efficiently~\cite{Bodlaender96,Kloks94,Korhonen21}. %($\star$)

\smallskip
\noindent \textbf{Integer Quadratic Programming.}\quad
\textsc{Integer Quadratic Programming} (IQP) over $d$ dimensions can be formalized as the task of computing
\begin{equation}\label{eq:generalIQP}\tag{IQP}
	\max \left\{ x^{T} Q x \mid A x \le b,\, x \ge 0 ,\, x \in \mathbb{Z}^d \right\} \,,
\end{equation}
where $Q \in \mathbb{Z}^{d \times d}$, $A  \in \mathbb{Z}^{m \times d}$, $b \in \mathbb{Z}^{m}$.
That is, IQP asks for an integral vector $x \in \mathbb{Z}^d$ which maximizes the value of a quadratic form subject to satisfying a set of linear constraints.

\begin{proposition}[{\cite{Lokshtanov15,Zemmer2017}}, see also~\cite{GavenciakKK22}]\label{prop:IQPisFPT}
	\textsc{Integer Quadratic Programming} is fixed-parameter tractable when parameterized by $d+\|A\|_{\infty}+ \|Q\|_{\infty}$.
\end{proposition}

%%%%%%%%%%%%%%%%%%%%%%%%%%%%%%%%%%%%%%%%%%%%%%%%%%%%%%%%%%%%%%
\section{Structural Properties of Outcomes}
\label{sec:structure}

As our first set of contributions, we establish some basic properties of our model and the associated problems that are studied within this paper. We begin by showcasing that the imposition of individual rationality or Nash stability as additional constraints on our outcomes does in fact have an impact on the maximum welfare that can be achieved (and hence it is indeed necessary to consider three distinct problems). We do not consider this to be obvious at first glance: intuitively, an agent $i$'s own contribution to the social welfare can only improve if they perform an IR- or NS-deviation, and the fact that the distance function $\dist_{\Pi_i}$ is symmetric would seem to suggest that this can only increase the total social welfare.

\begin{figure}
	\centering

	\captionbox{Social Network from \Cref{lem:irrational}.\label{fig:irrational}}[.48\textwidth]{
		\centering
		\begin{tikzpicture}
			[
			every node/.style={draw, fill=black, shape=circle, inner sep=0pt, text width=1.5mm, align=center, label distance=1mm}
			]
			\node[label=left:$x$] (x) at  (0,0) {};
			\node (p2)  at (1,0.5) {};
			\node (p4) at (1,-0.5) {};
			\node (p1) at (2,0.5) {};
			\node (p5) at(2,-0.5) {};
			\node (1) at  (4,1) {};
			\node (5) at (4,-1) {};
			\node (2) at (3.6,0.5) {};
			\node (4) at (3.6,-0.5) {};
			\node (3) at (3.4,0) {};
			\draw (p1) to (p2) to node[above=2mm,draw=none,fill=none] {$P$} (x) to (p4) to (p5);
			\draw (1) to (2) to (3) to (4) to (5) to node[right,draw=none,fill=none] {$K$} (1) to (3) to (5) to (2) to (4) to (1);
			\foreach \i in {1,...,5}
			{
				\draw (p1) to (\i) to (p5);
			}
		\end{tikzpicture}
	}
	\captionbox{Social Network from \Cref{lem:non-nash}.\label{fig:non-nash}}[.48\textwidth]{
		\begin{tikzpicture}
			[
			every node/.style={draw, fill=black, shape=circle, inner sep=0pt, text width=1.5mm, align=center, label distance=1mm}
			]
			\node[label=above:$x$] (x) at  (0,0) {};
			\node (p2)  at (1,0.5) {};
			\node (p4) at (1,-0.5) {};
			\node (p1) at (2,0.5) {};
			\node (p5) at(2,-0.5) {};
			\node (1) at  (4,0.7) {};
			\node (5) at (4,-0.7) {};
			\node (2) at (3.6,0.3) {};
			\node (4) at (3.6,-0.3) {};
			\node[label=above:$y$] (y) at (-1,0) {};
			\draw (p1) to (p2) to node[above=2mm,draw=none,fill=none] {$P$} (x) to (p4) to (p5);
			\draw (x) to (y);
			\draw (1) to (2) to (4) to (5) to node[right,draw=none,fill=none] {$K$} (1) to (4);
			\draw (2) to (5);
			\foreach \i in {1,2,4,5}
			{
				\draw (p1) to (\i) to (p5);
			}
		\end{tikzpicture}
	}
\end{figure}

\begin{lemma}\label{lem:irrational}
	There is a scoring vector $\val$ and a social network $G$ such that the single outcome achieving the maximum social welfare is not individually rational.
\end{lemma}
\begin{proof}
	Consider a scoring function $\val$ such that $\val=(1,1,-1,-1,-1,-1)$.
	Consider the social network $G$ in \Cref{fig:irrational} formed from a path $P$ on $5$ vertices and a clique $K$ on $5$ vertices by connecting the endpoints of~$P$ to all vertices of $K$.
	Let $x$ be the central agent of $P$.
	Let~$C$ be the grand coalition in $G$.
	The graph can be viewed as a $6$-cycle with $K$ forming one ``bold'' agent.
	All vertices on the cycle contribute positively to the agent's utility, except for the one that is exactly opposite on the cycle.
	Hence, $\util(x,C)=4-5=-1$, while utility of all other agents is $8-1=7$ in $C$.
	This gives total social welfare of $62$ for the grand coalition.

	However, if $x$ leaves the coalition to form its own one, their utility will improve from $-1$ to $0$, whereas the total social welfare drops.
	Indeed, in $C \setminus \{x\}$ there are 2 agents with utility $6-2=4$, 2 agents with utility $7-1=6$ and 5 agents with utility $8-0$, giving total social welfare of $60$.
	If any $y\neq x$ was to be excluded from $C$ to form outcome $\{y\}, C\setminus \{y\}$, then $y$ joining $C$ improves social welfare, proving that it was not optimal.
	Finally, if the outcome consists of several coalitions with the largest one of size 8, then the welfare is at most $8 \cdot 7+2 \cdot 1= 56$, if the largest size is 7, then we get at most $7 \cdot 6+3\cdot 2=48$, for 6 it is $6\cdot 5+4\cdot 3=42$ and for 5 it is $5 \cdot 4 +5 \cdot 4=40$.

	Hence the grand coalition $C$ is the only outcome with maximal social welfare, but it is not individually rational (and therefore not Nash stable), as $\util(x,C)=-1$.
\end{proof}

\begin{lemma}\label{lem:non-nash}
 	There is a scoring vector $\val$ and a social network $G$ such that the single individually rational outcome achieving the maximum social welfare among such outcomes is not Nash stable.
\end{lemma}
\begin{proof}
	Consider again the scoring function $\val=(1,1,-1,-1,-1,-1)$.
	Similarly to previous lemma, consider the social network $G$ in \Cref{fig:non-nash} formed from a path $P$ on $5$ vertices and a clique $K$ on $4$ vertices by connecting the endpoints of~$P$ to all vertices of $K$ and adding a agent $y$ only connected to the central agent of $P$ which we call $x$.
	Let~$C$ be the coalition containing all vertices of $G$ except for $y$.
	As in the previous lemma, $G[C]$ can be viewed as a $6$-cycle with $K$ forming one ``bold'' agent.
	Hence,~$\util_x(C)=4-4=0$, while utility of other agents in $C$ is $7-1=6$.
	Trivially $\util_y(\{y\})=0$, hence the outcome $(\{y\},C)$ is individually rational.
	It has total social welfare of $48$.
	However, it is not Nash stable, as $x$ wants to deviate to $\{x,y\}$ giving them utility $1$.

	However, the outcome $(\{x,y\}, C\setminus\{x\})$, which is Nash stable, has total social welfare only $46$.
	Note that $\util_z(C\setminus\{x\}) \ge 3$ for every agent $z \in C\setminus\{x\}$, so any outcome $(\{x,y,z\}, C\setminus\{x,z\})$ cannot be Nash stable.
	While the total social welfare of the grand coalition is $46$, the utility of $y$ is $3-6=-3$ in this coalition, so this outcome is not even individually rational.
	From the computations in the previous lemma, it follows, that to attain the social welfare of $48$, the largest coalition in the outcome must be of size at least \(7\).
	Moreover, if it is of size exactly \(7\), then these \(7\) vertices must be at mutual distance at most~$2$.
	However, there are no \(7\) vertices in mutual distance at most \(2\) in $G$.
	Hence, in any outcome with social welfare $48$ the largest coalition must be of size at least $8$.
	Agent $y$ has only \(3\) agents in distance at most \(2\) in $G$.
	Hence, for $y$ to get a positive utility from some coalition, the coalition must be of size at most \(7\), i.e., $y$ cannot be part of the largest coalition in any outcome with social welfare at least \(48\).
	However, for every $z \in C$, $z$ joining the coalition $C\setminus \{z\}$ improves the social welfare of the outcome, proving that it was not optimal.

	Hence the outcome $(\{y\},C)$ is the only individually rational outcome with maximal social welfare, but it is not Nash stable.
\end{proof}

It should be noted that \Cref{lem:irrational,lem:non-nash} also contrast many other models where outputs maximizing social welfare are stable for symmetric utilities~\cite{BogomolnaiaJ2002,BiloMM22,BullingerS2023}.

As our next two structural results, we prove that on certain SDGs it is possible to bound not only the diameter but also the size of each coalition in a welfare-maximum outcome. Notably, we establish such bounds for SDGs on bounded-degree networks and SDGs which have a simple scoring vector on a tree-like network. While arguably interesting in their own right, these properties will be important for establishing the fixed-parameter tractability of computing welfare-optimal outcomes in the next section.

\begin{lemma}\label{lem:maxdeg-coal-size}
	For every scoring vector $\val=(\maxval,\ldots,s_\delta)$, if $G$ is a graph of maximum degree $\Delta(G)$ and $C$ is a coalition of size more than $(\maxval+1) \cdot \Delta(G) \cdot (\Delta(G)-1)^{\cdiam-1}$, then for every $i \in C$ we have $\util(i,C) <0$.
\end{lemma}
\begin{proof}
	Let $i \in C$.
	There are at most $\Delta(G) \cdot (\Delta(G)-1)^{\cdiam-1}$ agents in distance at most $\cdiam$ from $i$.
	Each of these agents contributes at most $\maxval$ to $\util(i,C)$.
	Every other agent contributes at most $-1$.
	Hence, if there are more than $(\maxval+1) \cdot \Delta(G) \cdot (\Delta(G)-1)^{\cdiam-1}$ agents in $C$, then
	more than $\maxval \cdot \Delta(G) \cdot (\Delta(G)-1)^{\cdiam-1}$ of them have a negative contribution to $\util(i,C)$ and
	\begin{equation*}
		\util(i,C) < \maxval \cdot \Delta(G) \cdot (\Delta(G)-1)^{\cdiam-1}
		-1 \cdot  \maxval \cdot \Delta(G) \cdot (\Delta(G)-1)^{\cdiam-1} =0. \qedhere
	\end{equation*}
\end{proof}

\begin{lemma}\label{lem:degen-coal-size}
	Let $\val=(s_1,\ldots,s_\delta)$ be such that $s_2 < 0$.
	If $G$ is a graph of treewidth $\tw$ and $C$ is a coalition of size more than $2(\maxval+1) \cdot \tw + 1$, then $\sum_{i \in C}\util(i,C) <0$.
\end{lemma}
\begin{proof}
	Each agent adjacent to $i$ contributes $\maxval$ to $\util(i,C)$, whereas all the other agents contribute at most~$-1$.
	Since a  graph of treewidth $\tw$ is $\tw$-degenerate, there are $|E(G[C])| \le |C| \cdot \tw$ pairs of adjacent agents and $\binom{|C|}{2} - |E(G[C])|$ pairs of non-adjacent agents.
	We have
	\begin{align*}
		\sum_{i \in C}\util(i,C)
		&= \sum_{i,j \in C; i\neq j}\val\left(\dist(i,j)\right)\\
		&\le 2\left(\maxval\cdot \left|E\left(G[C]\right)\right| - \left(\binom{|C|}{2} - \left|E\left(G[C]\right)\right|\right)\right)\\
		&= 2\left((\maxval+1) \cdot  \left|E\left(G[C]\right)\right| - \binom{|C|}{2}\right)\\
		&\le 2(\maxval+1) \cdot |C| \cdot \tw - |C|(|C|-1)\\
		&=|C|\left(2(\maxval+1) \cdot \tw- (|C|-1)\right)\\
		&<|C|\left(2(\maxval+1) \cdot \tw -\left(2(\maxval+1) \cdot \tw+ 1-1\right)\right)=0. \qedhere
	\end{align*}
\end{proof}

%%%%%%%%%%%%%%%%%%%%%%%%%%%%%%%%%%%%%%%%%%%%%%%%%%%%%%%%%%%%%%
\section{Computing Optimal Outcomes}
\label{sec:algo}

\subsection{Intractability}
As our first step towards an understanding of the complexity of computing a welfare-optimal outcome in an SDG, we establish the \NP-hardness of \DGWF, \DGIR\ and \DGNS\ even for a very simple choice of $\val$.

\begin{theorem}
	\label{thm:NPh}
	Let $\val=(s_1)$ for any $s_1>0$.
	Then \DGWF, \DGIR\ and \DGNS are \NPh.
\end{theorem}
\begin{proof}[Proof Sketch]
	As our first step, we prove the \NP-hardness of the intermediate problem called \textsc{3-Colo\-ring Triangle Covered Graph  (3CTCG)} via an adaptation of a known reduction from \textsc{Not\-All\-Equal-3-SAT}~\cite[Theorem 9.8]{Papadimi94}:

	\begin{center}
		\begin{boxedminipage}{0.98 \columnwidth}
			\textsc{3-Coloring Triangle Covered Graph  (3CTCG)}\\[2pt]
			\begin{tabular}{l p{0.78 \columnwidth}}
				Input: & An undirected graph $G=(V,E)$ with $|V|=3n$ vertices such that $G$ contains a collection of $n$ mutually vertex disjoint triangles.\\
				Question: \hspace{-0.4cm} & Does $G$ have a 3-coloring?
			\end{tabular}
		\end{boxedminipage}
	\end{center}

	Next, we reduce \textsc{3CTCG} to our three problems via a single construction. Let $G$ be an instance of \textsc{3CTCG} with $3n$ vertices and $T_1, \ldots, T_n$ the corresponding collection of triangles.
	Let $\overline{G}$ be a complement of $G$, let $\maxval=\maxval(\val)$ and let $b=3n\maxval\cdot(n-1)$.
	To establish the \NP-hardness of \DGWF, it suffices to show that $G$ is a \YesI{} of \textsc{3CTCG} if and only if $\overline{G}$ admits an outcome with social welfare at least $b$; for the remaining two problems, we additionally show that such an outcome will furthermore be individually rational and Nash stable.
\end{proof}

\subsection{An Algorithm for Tree-Like Networks}
We complement Theorem~\ref{thm:NPh} by establishing that all three problems under consideration can be solved in polynomial time on networks of bounded treewidth---in other words, we show that they are \XP-tractable w.r.t.\ treewidth.
We first describe the ``baseline'' algorithm for solving \DGWF, and then prove that this may be adapted to also solve the other two problems by expanding on its records and procedures (see the appendix).

\begin{theorem}
	\label{thm:tw}
	For every fixed scoring vector~$\val$, the \DGWF, \DGIR, and \DGNS problems are in~\XP when parameterized by the treewidth of the social network~$G$.
\end{theorem}
\begin{proof}[Proof Sketch]
	% Proof notation.

	\newcommand{\Pttn}{\ensuremath{C}} % Partition
	\newcommand{\len}{\ensuremath{\operatorname{len}}}
	\newcommand{\Spaths}{\ensuremath{S}} % Shortest paths
	\newcommand{\Dvec}{\ensuremath{T}} % Distance vectors
	Our algorithm is based on leaf-to-root dynamic programming along a nice tree-de\-com\-po\-si\-ti\-on of the input social network with rather complicated structure. In each node $x$ of the tree-decomposition, we store a set $\mathcal{R}_x$ of partial solutions called \emph{records}. Each record realizes a single \emph{signature} which is a triple~$(\Pttn,\Spaths,\Dvec)$, where
	\begin{itemize}
		\item $\Pttn$ is a partition of bag agents into parts of coalitions; there are at most $\tw+1$ different coalitions intersecting $\beta(x)$ and, thus, at most ${tw^\Oh{\tw}}$ possible partitions of $\beta(x)$.
		\item $\Spaths$ is a function assigning each pair of agents that are part of the same coalition according to $\Pttn$ the shortest intra-coalitional path; recall that for fixed \val, the diameter of every coalition is bounded by a constant~$\cdiam$ and, therefore, there are~${n^\Oh{\cdiam} = n^\Oh{1}}$ possible paths for each pair of agents which gives us~${n^\Oh{\tw^2}}$ combinations in total.
		\item $\Dvec$ is a table storing for every coalition $P$ and every possible vector of distances to bag agents that are in $P$ the number of agents from $P$ that were already forgotten in some node of the tree-decomposition; the number of possible coalitions is at most $\tw+1$, the number of potential distance vectors is $\cdiam^{\tw+1} = 2^\Oh{\tw}$, and there are at most $n$ values for every combination of coalition and distance vector which leads to at most~${n^{2^\Oh{\tw}}}$ different tables~$\Dvec$.
	\end{itemize}
	The value of every record is a pair $(\pi,w)$, where $\pi$ is a partition of~$V^x$ such that $\SW(\pi) = w$ and $\pi$ witnesses that there is a partition of $V^x$ corresponding to the signature of the record, as described above. We store only one record for every signature -- the one with the highest social welfare. Therefore, in every node $x$, there are at most $n^{2^\Oh{\tw}}$ different records.

	Once the computation ends, we check the record in the root node~$r$ and based on the value of $w$, we return the answer; \Yes\ if $w\geq b$ and \No\ otherwise. Moreover, as $G^r=G$, the partition $\pi$ is also an outcome admitting social-welfare~$w$.
\end{proof}

\subsection{Fixed-Parameter Tractability}
A natural follow-up question to Theorem~\ref{thm:tw} %and~\ref{thm:twstable}
is whether one can improve these results to fixed-parameter algorithms. As our final contribution, we show that this is possible at least when dealing with simple scoring vectors, or on networks with stronger structural restrictions. To obtain both of these results, we first show that to obtain fixed-parameter tractability it suffices to have a bound on the size of the largest coalition in a solution (i.e., a welfare-optimal outcome).

\begin{theorem}
	\label{thm:tw+coal-sz}
	For every fixed scoring vector~$\val$, the variants of \DGWF, \DGIR, \DGNS  where we only consider outcomes consisting of coalitions of at most a prescribed size are \FPT parameterized by the treewidth of the network and the maximum coalition size combined.
\end{theorem}
%\iflong
\begin{proof}[Proof Sketch]
	Similar to the previous ones, we design a dynamic programming (DP) on a nice tree decomposition, albeit the procedure and records are completely different.

	Given a subset of agents $X \subseteq N$, let $\Pi = (\pi_1,\pi_2, \dots, \pi_\ell)$ be a partition of a set containing $X$ and some ``anonymous'' agents. We use \emph{$\topo(\Pi)$} to denote a set of graph topologies on $\pi_1, \pi_2, \dots, \pi_\ell$ given $X$. That is, $\topo(\Pi) = \{ \topo(\pi_1), \dots , \topo(\pi_\ell)\}$ where $\topo(\pi_i)$ is some graph on $|\pi_i|$ agents, namely $\pi_i \cap X$ and $|\pi_i \setminus X|$ ``anonymous'' agents, for each $i \in [\ell]$.
	The maximum coalition size of any welfare maximizing partition is denoted by $\sz$.
	Table, {\sf M}, contains an entry {\sf M}$[x, C, \topo(\Pi)]$ for every node $x$ of the tree decomposition,  each partition $C$ of $\beta(x)$, and each set of graph topologies $\topo(\Pi)$ given $\beta(x)$ where $\Pi$ is a partition of at most $\sz\cdot\tw$ agents. An entry of {\sf M} stores the maximum welfare in $G^x$ under the condition that the partition into coalitions satisfies the following properties.
	Recall that for a partition $P$ of agents and an agent $a$, we use $P_a$ to denote the coalition agent $a$ is part of in $P$.
	\begin{enumerate}
		\item \emph{$C$ and $\Pi$ are consistent}, i.e., the partition of the bag agents $\beta(x)$ in $G^x$ is denoted by $C$ and $C_a = \Pi_a \cap \beta(x)$ for each agent $a \in \beta(x)$.
		\item The coalition of agent $a \!\in \!\beta(x)$ in the graph $G^x$ is $\Pi_a$.
		\item \emph{$\topo(\Pi)$ is consistent with $G^x$} i.e., the subgraph of $G^x$ induced on the agents in coalition of $a$ is $\topo(\Pi_a)$, i.e., $G^x[\Pi_a] = \topo(\Pi_a)$.
	\end{enumerate}

	Observe that we do not store $\Pi$. We only store the topology of $\Pi$ which is a graph on at most $\sz\cdot \tw$ agents.

	We say an entry of {\sf M}$[x,C, \topo(\Pi)]$ is \emph{valid} if it holds that
	\begin{enumerate}
		\item \emph{$C$ and $\Pi$ are consistent}, i.e., $C_a = \Pi_a \cap \beta(x)$ for each agent $a\in \beta(x)$,
		\item Either $C_a = C_b$, or $C_a \cap C_b = \emptyset$ for each pair of agents $a,b \in \beta(x)$,
		\item \emph{$\topo(\Pi)$ is consistent with $G^x$ in $\beta(x)$}, i.e., for each pair of agents $a,b \in \beta(x)$ such that $\Pi_{a} = \Pi_{b}$, there is an edge $(a,b) \in \topo(\Pi_{a})$ if and only if $(a,b)$ is an edge in $G^x$.
	\end{enumerate}

	Once the table is computed correctly, the solution is given by the value stored in  {\sf M}$[r,C, \topo(\Pi)]$ where~$C$ is empty partition and $\topo(\Pi)$ is empty. Roughly speaking, the basis corresponds to leaves (whose bags are empty), and are initialized to store $0$. For each entry that is not valid we store $-\infty$. To complete the proof, it now suffices to describe the computation of the records at each of the three non-trivial types of nodes in the decomposition and prove correctness.
\end{proof}
%\fi

Similarly to \Cref{thm:tw}, we design a dynamic programming on a nice tree decomposition, albeit the procedure and records are completely different.

From \Cref{lem:degen-coal-size} it follows that  if $s_2 < 0$ and $\tw(G)$ is bounded, then the maximum coalition size  of a welfare maximizing outcome is bounded. Hence, using Theorem~\ref{thm:tw+coal-sz}  we get the following.
\begin{corollary}
	\DGNS, \DGIR, and \DGWF are fixed-parameter tractable parameterized by the treewidth $\tw(G)$ if $s_2 < 0$.
\end{corollary}

Turning back to general scoring vectors, we recall that \Cref{lem:maxdeg-coal-size} provided a bound on the size of the coalitions in a welfare-optimal outcome in terms of the maximum degree $\Delta(G)$ of the network $G$. Applying Theorem~\ref{thm:tw+coal-sz} again yields:

\begin{corollary}
	\DGNS, \DGIR, and \DGWF are fixed-parameter tractable parameterized by the treewidth $\tw(G)$ and the maximum degree $\Delta(G)$ of the social network.
\end{corollary}

As our final contribution, we provide fixed-parameter algorithms for computing welfare-optimal outcomes that can also deal with networks containing high-degree agents. To do so, we exploit a different structural parameter than the treewidth---namely the vertex cover number of $G$ ($\vc(G)$). We note that while the vertex cover number is a significantly more ``restrictive'' graph parameter than treewidth, it has found numerous applications in the design of efficient algorithms in coalition formation, including for other types of coalition games~
\cite{BiloFMM18,BodlaenderHJOOZ20,HanakaL22}.

\begin{theorem}
	\label{thm:WfIsNsFPTwrtVC}
	\DGNS, \DGIR, and \DGWF are fixed-parameter tractable parameterized by the vertex cover number $\vc(G)$ of the social network.
\end{theorem}
\begin{proof}[Proof Sketch]
	Let $k = \vc(G)$ and let $U$ be a vertex cover for~$G$ of size~$k$.
	Observe that in each solution there are at most $k$ non-singleton coalitions, since~$G$ has a vertex cover of size~$k$ and each coalition must be connected.
	Furthermore, the vertices of $G - U$ can be partitioned into at most $2^k$ groups according to their neighborhood in the set~$U$.
	That is, there are $n_W$ vertices in $G - U$ such that their neighborhood is~$W$ for some $W \subseteq U$; denote this set of vertices~$I_W$.

	We perform exhaustive branching to determine certain information about the structure of the coalitions in a solution---notably:
	\begin{enumerate}
		\item
		which vertices of $U$ belong to each coalition (i.e., we partition the set $U$); note that there are at most $k^k$ such partitions, and
		\item
		if there is at least one agent of $I_W$ in the coalition or not
		; note that there are at most $(2^{2^k})^k$ such assignments of these sets to the coalitions.
	\end{enumerate}
	We branch over all possible admissible options of the coalitional structure described above possessed by a hypothetical solution. The total number of branches is upper-bounded by a function of the parameter value~$k$ and thus for the problems to be in \FPT it suffices to show that for each branch we can find a solution (if it exists) by a fixed-parameter subprocedure.
	To conclude the proof, we show that a welfare-maximum outcome (which furthermore satisfies the imposed stability constraints) with a given coalitional structure can be computed by modeling this as an Integer Quadratic Program where $d+\|A\|_{\infty}+ \|Q\|_{\infty}$ are all upper-bounded by a function of $k$---such a program can be solved in \FPT time using \Cref{prop:IQPisFPT}.

	The (integer) variables of the program are $x^C_W$, which express the number of vertices from the set $I_W$ in the coalition with $C \subseteq U$; thus, we have $x^C_W \in \mathbb{Z}$ and $x^C_W \ge 1$.
	Let $\mathcal{C}$ be the considered partitioning of the vertex cover~$U$.
	We use $C \in \mathcal{C}$ for the set $C \subseteq U$ in the coalition and $C^+$ for the set $C$ and the guessed groups having at least one agent in the coalition.
	We require that the vertices of $G-U$ are also partitioned in the solution, i.e.,
	\begin{equation}\label{eq:WfIsNsFPTwrtVC:IQP:partition}
		\sum_{C \in \mathcal{C}} \sum_{W \in C^+} x^C_W = n_W \qquad \forall W \subseteq U.
	\end{equation}
	The quadratic objective expresses the welfare of the coalitions in the solution while the linear constraints ensure the stability of the outcome; for the latter, we rely on the fact that it is sufficient to verify the stability for a single agent from the group~$I_W$ in each coalition.
\end{proof}

%%%%%%%%%%%%%%%%%%%%%%%%%%%%%%%%%%%%%%%%%%%%%%%%%%%%%%%%%%%%%%
\section{Conclusions and Future Research Directions}

In this work, we studied social distance games through the lens of an adaptable, non-normalized scoring vector which can capture the positive as well as negative dynamics of social interactions within coalitions.
The main focus of this work was on welfare maximization, possibly in combination with individual-based stability notions---individual rationality and Nash stability.
It is not surprising that these problems are intractable for general networks; we complement our model with algorithms that work well in tree-like environments.

Our work opens up a number of avenues for future research.
One can consider other notions of individual-based stability such as individual stability~\cite[pp.~360--361]{BrandtCELP16}\cite{GanianHKSS22}, or various notions of group-based stability such as core stability~\cite[p.~360]{BrandtCELP16}\cite{BranzeiL2011,OhtaBISY17}.
Furthermore, our results do not settle the complexity of finding stable solutions (without simultaneous welfare maximization).
Therefore, it remains open if one can find a Nash stable solution for a specific scoring vector.
Also, a more complex open problem is to characterize those scoring vectors that guarantee the existence of a Nash (or individually) stable solution.

Finally, we remark that the proposed score-based \DGs\ model can be generalized further, e.g., by allowing for a broader definition of the scoring vectors. For instance, it is easy to generalize all our algorithms to scoring vectors which are not monotone in their ``positive part''. One could also consider situations where the presence of an agent that is ``far away'' does not immediately set the utility of other agents in the coalition to $-\infty$. One way to model these settings would be to consider ``\emph{open}'' scoring vectors, for which we set $\val(a)=\val(\cdiam)$ for all $a>\cdiam$---meaning that distances over $\cdiam$ are all treated uniformly but not necessarily as unacceptable.

Notice that if $\val(\cdiam) \geq 0$ for an open scoring vector $\val$, the grand coalition is always a social-welfare maximizing outcome for all three problems---hence here it is natural to focus on choices of $\val$ with at least one negative entry. We note that all of our fixed-parameter algorithms immediately carry over to this setting for arbitrary choices of open scoring vectors $\val$. The situation becomes more interesting when considering the small-world property: while the diameter of every welfare-maximizing outcome can be bounded in the case of Nash stable or individually rational coalitions (as we prove in our final Theorem~\ref{thm:diam-nashir} below), whether the same holds in the case of merely trying to maximize social welfare is open and seems to be a non-trivial question. Because of this, Theorem~\ref{thm:tw} can also be extended to the \DGIR and \DGNS with open scoring vectors, but it is non-obvious for \DGWF.

\begin{theorem}
	\label{thm:diam-nashir}
	Let $\val=(s_1,\dots,s_{\cdiam})$ be an arbitrary open scoring vector and $G$ be a social network. Every outcome $\Pi$ containing a coalition $C\in\Pi$ with diameter exceeding \(\ell = 2\cdot\maxval\cdot \cdiam\) can be neither Nash-stable nor individually rational.
\end{theorem}
\begin{proof}[Proof Sketch]
	Consider a shortest path $P$ in $C$ whose length exceeds $\ell$. We identify a set of edge cuts along $P$ and show that at least one such cut must be near an agent whose utility in $C$ is negative, due to the presence of a large number of agents that must be distant from the chosen edge cut.
\end{proof}

%%%%%%%%%%%%%%%%%%%%%%%%%%%%%%%%%%%%%%%%%%%%%%%%%%%%%%%%%%%%%%
\paragraph*{Acknowledgements.}
All authors are grateful for support from the OeAD bilateral Czech-Austrian WTZ-funding Programme (Projects No. CZ 05/2021 and 8J21AT021). Robert Ganian acknowledges support from the Austrian Science Foundation (FWF, project Y1329). Thekla Hamm also acknowledges support from FWF, project J4651-N. Dušan Knop, Šimon Schierreich, and Ondřej Suchý acknowledge the support of the Czech Science Foundation Grant No.~22-19557S. Šimon Schierreich was additionally supported by the Grant Agency of the Czech Technical University in Prague, grant \mbox{No.~SGS23/205/OHK3/3T/18}.

\bibliographystyle{eptcs}
\bibliography{references}

\end{document}